\documentclass{amsart}
\usepackage{graphicx}
\usepackage{amsthm,amsmath,verbatim,amssymb}
\usepackage{color}
\usepackage{arydshln}
\newtheorem{lem}{Lemma}

\newtheorem{prop}{Proposition}
\newtheorem{thm}{Theorem}

\DeclareMathOperator{\Tr}{Tr}

\newcommand{\p}{\psi}
\newcommand{\la}{\lambda} 
\newcommand{\e}{\mathbb{E}}

\newcommand{\R}{\mathbb{R}}

\begin{document}
\title[Spectrum]{%
Spectrum of SYK model III:\\ Large deviations and concentration of measures}

\author[Feng, Tian and Wei]{Renjie Feng, Gang Tian, Dongyi Wei}
\address{Beijing International Center for Mathematical Research, Peking University, Beijing, China, 100871.}

\email{renjie@math.pku.edu.cn}
\email{gtian@math.pku.edu.cn}
\email{jnwdyi@pku.edu.cn}

\date{\today}
   \maketitle
    \begin{abstract}
 In \cite{FTD1}, we proved the almost sure convergence of eigenvalues of the SYK model, which can be viewed as a type of \emph{law of large numbers} in probability theory; in  \cite{FTD2}, we proved that the linear statistic of eigenvalues satisfies the \emph{central limit theorem}. In this article, we continue to study another important theorem in probability theory\,-- the \emph{concentration of measure theorem}, especially for the Gaussian SYK model.  
 We will  prove a \emph{large deviation principle} (LDP)  for the normalized empirical measure of eigenvalues when $q_n=2$,  
 in which case the eigenvalues can be 
 expressed in term of these of Gaussian random antisymmetric matrices. Such LDP result has its own independent interest in  random matrix theory. For general $q_n\geq 3$,  we can not prove the LDP, we will prove a concentration of measure theorem   by estimating the Lipschitz norm of the  Gaussian SYK model.
\end{abstract}

\section{Introduction}

 In this article, we will study the large deviation principle and the concentration of measure theorem for the Gaussian SYK model, instead of the general SYK model considered in \cite{FTD1, FTD2}. 
 
 The Gaussian SYK model  is \cite{black, wenbo susy, K, MS, SY} \begin{equation}\label{sykm}H=i^{[q_n/2]}\frac{1}{\sqrt{{{n} \choose {q_n}}}}\sum_{1\leq i_1<i_2< \cdots < i_{q_n} \leq n} J_{i_1i_2\cdots i_{q_n}}{\p}_{i_1}{\p}_{i_2}\cdots {\p}_{ i_{q_n}},\end{equation}
where $n$ is an even integer, $J_{i_1i_2\cdots  i_{q_n}} $ are independent identically distributed
(i.i.d.) standard real Gaussian random variables with mean 0 and variance 1; 
$\p_j$ are Majorana fermions  satisfying the algebra
\begin{equation}\label{anti}\left\{\p_i,\p_j\right\}:=\p_i\p_j+\p_j\p_i=2\delta_{ij},\,\,\,1\leq i, j\leq n  .\end{equation}

By the representation of the Clifford algebra, $\p_i$ can be represented by $L_n\times L_n$ Hermitian matrices with $L_n=2^{n/2}$. Actually $\{\psi_i\}_{1\leq i\leq n}$ can be generated by Pauli matrices  iteratively \cite{lo}.  
Let  $\la_{i}, 1\leq i \leq L_n$ be the eigenvalues of $H$. One may check that $H$ is Hermitian by the anticommunitative relation \eqref{anti}, thus $\la_{i}$ are real numbers.  One of the main tasks in random matrix theory is to understand 
the  following normalized empirical measure of eigenvalues of $H$  	\begin{equation}\label{emp}
  \rho_n(\lambda):=\frac 1{L_n}\sum_{i} \delta_{\la_{i}}(\la).
	\end{equation}
Let's first summarize the main results in \cite{FTD1, FTD2}. Other than the standard Gaussian random variables, in  \cite{FTD1}, we consider the general cases where $J_{i_1i_2\cdots i_{q_n}}$ are i.i.d. random variables with mean 0 and variance 1, and the $k$-th moment  of $|J_{i_1i_2\cdots  i_{q_n}}| $ is uniformly bounded for any fixed $k$. 
We proved that $\rho_n$  converges to a probability measure $\rho_\infty$ almost surely (or with probability 1) in the sense of distribution, and the limiting density $\rho_\infty$ depends on the limit of the quotient ${q_n^2}/n$.  To be more precise, let $2\leq q_n\leq n/2$ be even, then $\rho_\infty$ will be the standard Gaussian measure if ${q_n^2}/n\to 0$; $\rho_\infty$ is the semicircle law if ${q_n^2}/n\to \infty$; and $\rho_\infty$ is related to the $q$-Hermite polynomial theory if ${q_n^2}/n\to a$. The results can be extended to even $q_n\geq n/2$ immediately.   One can also derive the results for $q_n$ odd.  The main result in \cite{FTD2} is that the linear statistic of eigenvalues satisfies the central limit theorem, which indicates the information about the 2-point correlation of the eigenvalues.   Regarding the spectral properties of the SYK model, we also refer to the numerical results in \cite{GV, GV2, GV3, GV4}. 


In this article, we continue to study the spectrum of the Gaussian SYK model. We will prove a large deviation principle (LDP) for eigenvalues when $q_n=2$ and  a concentration of measure theorem for general $q_n\geq 3$.

Throughout the article, we always assume $n$ is an even integer, $J_{i_1\cdots i_{q_n}}$ are standard Gaussian random variables and $q_n^2/n$ has a limit.   In physics, people care especially when $q_n$ is an  even integer, but  the model is still a good one in mathematics if $q_n$ is odd. Our main results apply to both cases. Moreover,  we only state and prove the main results for $0<q_n\leq n/2$, the results can be extended to $q_n\geq n/2$   immediately. This is because, as explained in \cite{FTD1}, there is a symmetry between the systems  with interaction of $q_n$ fermions   and $n-q_n$ fermions.



\subsection{Large deviations}\label{ldppi}
When $q_n=2$,   the SYK model reads
\begin{equation}\label{22}H=\frac{i}{\sqrt{{{n} \choose {2}}}}\sum_{1\leq i_1<i_2\leq n} J_{i_1i_2}{\p}_{i_1}{\p}_{i_2}.\end{equation}
 Let \begin{equation}\label{syk2}J=(J_{ij})_{1\leq i,j\leq n},\ J_{ji}:=-J_{ij} \end{equation} be the real Gaussian antisymmetric matrices.  This system is totally solvable in physics. If the eigenvalues of $J$ are $\pm i\mu_j$ where  $\mu_j\geq 0$ for $ 1\leq j\leq n/2$, then all eigenvalues of $H$ are given explicitly as  \cite{black, FTD1, GV3, MS}
\begin{equation}\label{dsdd}{{n} \choose {2}}^{-\frac{1}{2}}\sum\limits_{j=1}^{n/2} \pm \mu_{j}.\end{equation}
The normalized empirical measure defined in \eqref{emp} reads
\begin{equation}\label{2222}
	\rho_n:=\frac 1{L_n}\sum_{a_1,\cdots,a_{n/2}\in \{\pm 1\}} \delta_ {{{n}\choose {2}}^{-\frac{1}{2}}\sum\limits_{j=1}^{n/2} a_j \mu_{j}}.
\end{equation}
Then $\rho_n$ will tend to the standard Gaussian measure almost surely \cite{FTD1} and the linear statistic of these eigenvalues satisfies the central limit theorem \cite{FTD2}. 
In this article, we will further study its large deviation principle. We refer to  \cite{AGZ}  for the definition and basic properties of the LDP, and several well-known LDP results regarding the eigenvalues of random matrices. 

To state our result, we need to introduce an  auxiliary space. 
Let $X$ be a subspace of $l^{\infty},$
\begin{equation}\label{x}X=\left\{(x_j)_{j=0}^{\infty}\in l^\infty|x_j\geq x_{j+1}\geq 0,\ \forall\ j\in\mathbb{Z},j>0,\ x_0\geq \sum_{j=1}^{+\infty}x_j^2\right\},\end{equation} where \begin{equation}l^{\infty}=\{(x_j)_{j=0}^{\infty}|x_j\in\mathbb{R};\ \sup\limits_{j\geq 0}|x_j|<+\infty\},\end{equation} with the metric \begin{equation}d(x,y)=\sup\limits_{j\geq 0}|x_j-y_j|,\end{equation}
 for $x=(x_j)_{j=0}^{\infty},\ y=(y_j)_{j=0}^{\infty}.$ Then $(l^{\infty},d) $ is a complete metric space. By Fatou's Lemma, we know that $X$ is a closed subspace of $l^{\infty},$ thus $(X,d) $ is also a complete metric space (Polish space).
 
For $n$ even,  let us define $\gamma_n\in X$  as $(\gamma_n)_j={n\choose 2}^{-\frac{1}{2}}\mu_j $  for $1\leq j\leq n/2,$ $(\gamma_n)_0={n\choose 2}^{-1}\sum\limits_{j=1}^{n/2}\mu_j^2 $ and $(\gamma_n)_j=0 $  for $ j> n/2, $ i.e., 
 \begin{equation}\label{ln}\gamma_n=\left ({n\choose 2}^{-1}\sum\limits_{j=1}^{n/2}\mu_j^2, {n\choose 2}^{-\frac{1}{2}}\mu_1,\cdots, {n\choose 2}^{-\frac{1}{2}}\mu_{n/2},\, 0,\,\cdots \right).\end{equation}
For $x=(x_j)_{j=0}^{\infty}\in X, $ let \begin{equation}\label{jfunction}J(x):=x_0- \sum_{j=1}^{+\infty}x_j^2\end{equation} and \begin{equation}X_0=\{x\in X|J(x)=0\},\end{equation} then we have \begin{equation}J(x)\geq 0\,\,\mbox{and}\,\,\,\gamma_n\in X_0.\end{equation}
  We first have the LDP of $(\gamma_{n})_{n>0,n\in 2\mathbb{Z}}$  in this auxiliary space, 
   \begin{prop} \label{ldpp}
Let $\pm i\mu_j$  be eigenvalues of Gaussian antisymmetric matrices $J $ as in \eqref{syk2}.
Then the  random measure $(\gamma_{n})_{n>0,n\in 2\mathbb{Z}}$ defined in \eqref{ln} satisfies the LDP in $(X,d)$ with speed $n^2/4$ and good rate function \begin{equation}\label{i}I(x)=\begin{cases}x_0-1-\ln J(x), \,\,\, x\notin X_0; \\ +\infty, \,\,\,\,\,\,\,\,\quad\quad\quad\quad \,\, x\in X_0.\end{cases}\end{equation}
  \end{prop}
We define $\ln 0=-\infty$, then $I$ is lower semicontinuous by Fatou's lemma. As $J(x)=x_0- \sum\limits_{j=1}^{+\infty}x_j^2\leq x_0,$ we have $I(x)=x_0-1-\ln J(x)\geq J(x)-1-\ln J(x)\geq 0$. If the equality holds, we must have $x_0=J(x)=1$ and $\sum\limits_{j=1}^{+\infty}x_j^2=0$, i.e., $x_j=0$ for $j>0$; actually this is the only point where $I(x)$ achieves its minimum, i.e., \begin{equation}\label{min}I(x_{min})=0, \,\, \,\,x_{min}=(1,0,\cdots).\end{equation}

 Let  $M_1(\mathbb{R}) $ be the set of Borel probability
measures on $\mathbb{R}$ equipped  with the bounded Lipschitz metric \begin{equation}\label{md}d_{BL}(\mu,\nu)=\sup|\langle  \mu, f \rangle-\langle  \nu, f \rangle|,\end{equation}
 where the supremum is subject to all 1-Lipschitz functions $f : \mathbb{R}\to \mathbb{R}$, i.e., $$\ |f(x)-f(y)|\leq |x-y|\,\,\,\,\mbox{and}\,\,\,\,|f(x)|\leq 1.$$ Then $(M_1(\mathbb{R}), d_{BL})$ is a Polish space \cite{AGZ}. 

The LDP of the normalized empirical measure \eqref{2222} in $(M_1(\mathbb{R}), d_{BL})$ will be induced by the LDP of $(\gamma_{n})_{n>0,n\in 2\mathbb{Z}}$ in $(X, d)$, where we need to construct a continuous and injective function \begin{equation} \varphi:X\to M_1(\mathbb{R})\end{equation} such that $ \varphi(\gamma_n)=\rho_n.$ By \eqref{2222}, the Fourier transform of $\rho_n $ is
\begin{equation}\label{rhot}\widehat{\rho_n}(s)=\langle  \rho_n(\lambda), e^{is\lambda} \rangle =\prod_{j=1}^{n/2}\cos {{n} \choose {2}}^{-\frac{1}{2}}s\mu_j.\end{equation} If we define the Fourier transform of the measure $\varphi$ as \begin{equation}\label{definf} \widehat{\varphi(x)}(s)=e^{-J(x)s^2/2}\prod_{j=1}^{+\infty}\cos sx_j,\,\,\,\,x=(x_j)_{j=0}^{\infty}\in X,\end{equation}
then by definition of  $\gamma_n\in X_0$, we must have \begin{equation}\varphi(\gamma_n)=\rho_n.\end{equation} 
In \S \ref{ssss}, we will further show that $ \varphi$ is a Borel probability measure, continuous and injective. Hence, by the Contraction Principle (cf. Theorem D.7 in \cite{AGZ}), we have  
\begin{thm} \label{ldp2}
 The  normalized empirical measure $\rho_n$ \eqref{2222} of eigenvalues of the Gaussian SYK model for $q_n=2$ satisfies the LDP in $(M_1(\mathbb{R}),d_{BL})$ with speed $n^2/4$ and good rate function $\widetilde{I}$ such that $\widetilde{I}(x)=I(\varphi^{-1} x)$ if $x\in\varphi(X)$ and $\widetilde{I}(x)=+\infty$ if $x\not\in\varphi(X),$ where $I(x)$ is defined by \eqref{i}.
  \end{thm}
As a remark,  by \eqref{min}, one can conclude that $\widetilde{I}$ will achieve  its minimum at $\varphi(x_{min})$ where $x_{min}=(1,0,\cdots)$. By definition \eqref{definf}, the Fourier transform $\widehat{\varphi(x_{min})}(s)$ is  the  Gaussian function,  and thus  $\varphi(x_{min})$ is the Gaussian distribution, which implies that  $\widetilde{I}$  achieves  its minimum at the Gaussian distribution. 
 

 \subsection{Concentration of measure theorem}
 We can not derive the LDP for general $q_n\geq 3$, but we can prove a weaker version which is the concentration of measure theorem. The proof is based on the following classical Gaussian concentration of measure theorem \cite{Led}: 
 Let $(a_k)_{1\leq k\leq N}$ be $N$-dimensional Gaussian random vectors, and let $F:\R^{N}\to\R $ be Lipschitz with Lipschitz constant $L$, then there are universal constants $C,c>0$ such that for $t>0,$
\begin{equation}\label{P}\mathbb{P}[|F(a_1,\cdots, a_N)-\e F(a_1,\cdots, a_N)|>t]\leq Ce^{-ct^2/L^2}.
\end{equation}
We denote the set   $$I_n=\{(i_1, i_2,\cdots, i_{q_n}),\,\,1\leq i_1<i_2< \cdots <i_{q_n} \leq n \}.$$
 For any coordinate $R=(i_1, \cdots, i_{q_n})\in I_n$, we denote $$J_R:=J_{i_1\cdots i_{q_n}}\,\,\,\mbox{and}\,\,\,\,\Psi_R:=\psi_{i_1}\cdots \psi_{i_{q_n}}. $$  
Then we can simply rewrite \begin{equation}\label{simple}H=\frac{i^{[q_n/2]}}{\sqrt{{{n} \choose {q_n}}}}\sum_{R\in I_n} J_{R}\Psi_R.\end{equation}
 If we consider 
$H$ as a function of the standard Gaussian random vectors $(J_R)_{R\in I_n}$, then we first have the following Lipschitz estimates.
\begin{lem}\label{lemlip}Let $x:=(J_R)_{R\in I_n}\in \R^{{n\choose q_n}}$ be the Gaussian random vector and  $\rho_n$ be the normalized empirical measure \eqref{emp}. We consider the SYK model $H:=H(x)$ as a function of $x$. 
\\ (a) Let $f:\R\to\R$ be Lipschitz, then the map $x\mapsto\langle f,   \rho_n\rangle$ is ${n\choose q_n}^{-1/2} \|f'\|_{L^\infty(\mathbb R)}$-Lipschitz; \\(b) For any probability measure $ \rho$ on $\R$, the map $x\mapsto d_{BL} ( \rho_n,\rho)$ is ${n\choose q_n}^{-1/2} $-Lipschitz. 
\end{lem}
 Once we have the above Lipschitz estimates, by \eqref{P} of the classical concentration of measure theorem for Gaussian random vectors, we can prove
 
\begin{thm}\label{cmt}
Let $\rho_n$ be the normalized empirical measure of the Gaussian SYK model for any $0<q_n\leq n/2$ as in \eqref{emp}  and $\rho_\infty$ be the   limiting measure according to the limit $q^2_n/n$ as we derived in \cite{FTD1}. Given $a>0$, then there exists $C(a)>0$ such that $$\mathbb{P}(d_{BL}(\rho_n,\rho_{\infty})>a)\leq Ce^{-c(a){n\choose q_n}},$$
where $C$ is some universal constant. 
\end{thm}

\textbf{Acknowledgement:} The first named author would like to thank Gerard Ben Arous for many helpful discussions when he was visiting NYU Shanghai.

\section{Large deviation principle for $q_n=2$}
When $q_n=2$,  the system is totally solvable and all eigenvalues can be expressed in term of eigenvalues of Gaussian random antisymmetric matrices (see \eqref{dsdd}). In this section, we will prove the LDP for the normalized empirical measure $\rho_n$ (which is defined in \eqref{2222}) of these eigenvalues. There are mainly two steps: we will first  derive the LDP in an auxiliary space $(X, d)$,  then we construct a continuous and injective map $\varphi: X\to  M_1(\mathbb R)$  which will induce the LDP in $( M_1(\mathbb R), d_{BL})$ by the Contraction Principle.
\subsection{Some integral inequalities}
  Let $J$ be the real Gaussian antisymmetric matrices as in \eqref{syk2}.  We assume the eigenvalues of $J$ are $\pm i\mu_j$ where  $\mu_j\geq 0$ for $ 1\leq j\leq n/2$. Then the joint density of these eigenvalues is \cite{M}
\begin{align*} J_n(\mu):=\frac{1}{Z_n}|\Delta(\mu)|^2e^{-\sum\limits_{j=1}^{n/2}\mu_j^2/2}1(\mu_1>\cdots>\mu_{n/2}>0),
\end{align*}where $$\Delta(\mu)=\prod\limits_{1\leq i<j\leq n/2}\left(\mu_i^2-\mu_j^2\right),\,\,\,\mu:=(\mu_1,\cdots,\mu_{n/2}) $$ is the Vandermonde determinant. By Selberg
integrals, the normalization constant \begin{align*} {Z_n}=(\pi/2)^{\frac{n}{4}}\prod\limits_{j=0}^{ n/2-1}(2j)!.
\end{align*}
Given $$x:=(x_1,\cdots,x_{n/2}), $$
 let's denote $$x_{>k}:=(x_{k+1},\cdots,x_{n/2}), \,\,\, \Delta(x_{>k}):=\prod\limits_{k< i<j\leq n/2}(x_i^2-x_j^2)$$ and $$\Sigma_{n-2k}\\:=\{(x_{k+1},\cdots,x_{n/2}):x_{k+1}>\cdots>x_{n/2}>0\}$$
  for $0\leq k\leq n/2$.  Then for $x\in\Sigma_{n}$,  we have $$x=x_{>0},\ 0<\Delta(x_{>k-1})=\Delta(x_{>k})\prod\limits_{j=k+1}^{ n/2}(x_k^2-x_j^2)<x_k^{n-2k}\Delta(x_{>k})$$ and $$0<\Delta(x)\leq \Delta(x_{>k})\prod\limits_{j=1}^{ k}x_j^{n-2j}. $$
      We will need several integral inequalities.\begin{lem}\label{lemma5}If $a,b<1/2$ and $0\leq k\leq n/2$,  we have
\begin{equation}\label{12}\mathbb Ee^{a\sum\limits_{j=1}^k\mu_j^2+b\sum\limits_{j=k+1}^{n/2}\mu_j^2}\leq 2^{nk}(1-2a)^{-k(n-k-\frac{1}{2})}(1-2b)^{-(\frac{n}{2}-k)(\frac{n-1}{2}-k)}.\end{equation} When $k=1, b=0, a=1/4$, we further have
 \begin{equation}\label{23}\int_{\Sigma_{n}}|\Delta(\mu)|^2e^{\mu_{1}^2/4
-\sum\limits_{j=1}^{n/2}\mu_j^2/2}d\mu\leq 2^{2n}{Z_{n}}.\end{equation}
\end{lem}\begin{proof} By definition we have\begin{align*} \int_{\Sigma_{n}}|\Delta(\mu)|^2e^{-\sum\limits_{j=1}^{n/2}\mu_j^2/2}d\mu={Z_n},
\end{align*} where $d\mu$ is the Lebesgue measure.  For $a>0$, by changing of variables, we have\begin{align*} \int_{\Sigma_{n}}|\Delta(\mu)|^2e^{-a\sum\limits_{j=1}^{n/2}\mu_j^2/2}d\mu={Z_n}a^{-\frac{n}{4}-2{n/2\choose 2}}={Z_n}a^{-\frac{n(n-1)}{4}}.
\end{align*}Therefore, let's denote $m:=n-2k,$ we have
\begin{align*} &Z_n\mathbb Ee^{a\sum\limits_{j=1}^k\mu_j^2+b\sum\limits_{j=k+1}^{n/2}\mu_j^2}\\=&
\int_{\Sigma_{n}}|\Delta(\mu)|^2e^{a\sum\limits_{j=1}^k\mu_j^2+b\sum\limits_{j=k+1}^{n/2}\mu_j^2
-\sum\limits_{j=1}^{n/2}\mu_j^2/2}d\mu\\ \leq& \int_{\Sigma_{n}}\prod\limits_{j=1}^{ k}\mu_j^{2(n-2j)}|\Delta(\mu_{>k})|^2e^{-(1-2a)\sum\limits_{j=1}^k\mu_j^2/2-(1-2b)\sum\limits_{j=k+1}^{n/2}\mu_j^2/2}d\mu\\ \leq & \prod\limits_{j=1}^{ k}\int_{\mathbb{R}_+}\mu_j^{2(n-2j)}e^{-(1-2a)\mu_j^2/2}d\mu_j\cdot\int_{\Sigma_{n-2k}}|\Delta(\mu_{>k})|^2e^{-(1-2b)
\sum\limits_{j=k+1}^{n/2}\mu_j^2/2}d\mu_{>k}\\ =&\left[ \prod\limits_{j=1}^{ k}(2^{n-2j-\frac{1}{2}}(1-2a)^{-(n-2j)-\frac{1}{2}}\Gamma(n-2j+\frac{1}{2}))\right]\cdot{Z_{n-2k}}(1-2b)^{-\frac{m(m-1)}{4}}\\ \leq& (1-2a)^{-nk+k(k+1)-\frac{k}{2}}\left[\prod\limits_{j=1}^{ k}(2^{n-2}\Gamma(n-2j+1))\right]\cdot{Z_{n-2k}}(1-2b)^{-\frac{m(m-1)}{4}}\\ =& (1-2a)^{-k(n-k-\frac{1}{2})}\left[\prod\limits_{j=1}^{ k}(2^{n-2}(\pi/2)^{-\frac{1}{2}}\frac{Z_{n-2j+2}}{Z_{n-2j}})\right]\cdot{Z_{n-2k}}(1-2b)^{-\frac{m(m-1)}{4}}\\ \leq& (1-2a)^{-k(n-k-\frac{1}{2})}2^{(n-2)k}{Z_{n}}(1-2b)^{-(\frac{n}{2}-k)(\frac{n-1}{2}-k)},
\end{align*}which further gives $$\label{1} \mathbb Ee^{a\sum\limits_{j=1}^k\mu_j^2+b\sum\limits_{j=k+1}^{n/2}\mu_j^2}\leq 2^{nk}(1-2a)^{-k(n-k-\frac{1}{2})}(1-2b)^{-(\frac{n}{2}-k)(\frac{n-1}{2}-k)}.
$$For $k=1,\ b=0,\ a=1/4$, we obtain\begin{align*} &\int_{\Sigma_{n}}|\Delta(\mu)|^2e^{\mu_{1}^2/4
-\sum\limits_{j=1}^{n/2}\mu_j^2/2}d\mu \leq (1-2/4)^{-(n-\frac{3}{2})}2^{n}{Z_{n}}\leq 2^{2n}{Z_{n}},
\end{align*}which completes the proof. \end{proof}
Let's denote the subset $$ \Sigma_{n,a,b}=\left\{(x_{1},\cdots,x_{n/2})\in \Sigma_{n}:a{n\choose 2}<\sum\limits_{j=1}^{n/2}x_j^2<b{n\choose 2}\right\}.$$
 
\begin{lem}\label{lemma6} For $0<a<1<b,$ we have,  \begin{align*} \int_{\Sigma_{n,a,b}}|\Delta(\mu)|^2e^{-\sum\limits_{j=1}^{n/2}\mu_j^2/2}d\mu\geq
{Z_n}\left(1-(ae^{1-a})^{\frac{n(n-1)}{4}}-(be^{1-b})^{\frac{n(n-1)}{4}}\right).\end{align*}
\end{lem}\begin{proof}  For $0<a<1<b,$ we have \begin{align*} &\int_{\Sigma_{n}\setminus\Sigma_{n,0,b}}|\Delta(\mu)|^2e^{-\sum\limits_{j=1}^{n/2}\mu_j^2/2}d\mu\\ \leq&
\int_{\Sigma_{n}\setminus\Sigma_{n,0,b}}|\Delta(\mu)|^2e^{-b^{-1}\sum\limits_{j=1}^{n/2}\mu_j^2/2}e^{-(1-b^{-1}) b{n\choose 2}/2}d\mu \\ \leq&
\int_{\Sigma_{n}}|\Delta(\mu)|^2e^{-b^{-1}\sum\limits_{j=1}^{n/2}\mu_j^2/2}e^{-(1-b^{-1}) b{n\choose 2}/2}d\mu \\=&b^{\frac{n(n-1)}{4}}\int_{\Sigma_{n}}|\Delta(\mu)|^2e^{-\sum\limits_{j=1}^{n/2}\mu_j^2/2}e^{-(b-1) {n\choose 2}/2}d\mu\\=&b^{\frac{n(n-1)}{4}}Z_n  e^{-(b-1)\frac{n(n-1)}{4}}={Z_n}(be^{1-b})^{\frac{n(n-1)}{4}},
\end{align*}here, we used the fact that $1-b^{-1}>0$ and \begin{align*} -\sum\limits_{j=1}^{n/2}\mu_j^2/2=-b^{-1}\sum\limits_{j=1}^{n/2}\mu_j^2/2-(1-b^{-1}) \sum\limits_{j=1}^{n/2}\mu_j^2/2\\ \leq-b^{-1}\sum\limits_{j=1}^{n/2}\mu_j^2/2-(1-b^{-1}) b{n\choose 2}/2
\end{align*} for $\mu\in \Sigma_{n}\setminus\Sigma_{n,0,b}$.  Similarly, \begin{align*} \int_{\Sigma_{n,0,a}}|\Delta(\mu)|^2e^{-\sum\limits_{j=1}^{n/2}\mu_j^2/2}d\mu\leq
\int_{\Sigma_{n}}|\Delta(\mu)|^2e^{-a^{-1}\sum\limits_{j=1}^{n/2}\mu_j^2/2}e^{-(1-a^{-1}) a{n\choose 2}/2}d\mu \\=a^{\frac{n(n-1)}{4}}{Z_n}e^{-(a-1)\frac{n(n-1)}{4}}={Z_n}(ae^{1-a})^{\frac{n(n-1)}{4}}.
\end{align*}Therefore, we will finish the proof by observing the following identity, 
 \begin{align*} Z_n=\int_{\Sigma_{n}}|\Delta(\mu)|^2e^{-\sum\limits_{j=1}^{n/2}\mu_j^2/2}d\mu=
\int_{\Sigma_{n}\setminus\Sigma_{n,0,b}}|\Delta(\mu)|^2e^{-\sum\limits_{j=1}^{n/2}\mu_j^2/2}d\mu\\
+\int_{\Sigma_{n,a,b}}|\Delta(\mu)|^2e^{-\sum\limits_{j=1}^{n/2}\mu_j^2/2}d\mu
+\int_{\Sigma_{n,0,a}}|\Delta(\mu)|^2e^{-\sum\limits_{j=1}^{n/2}\mu_j^2/2}d\mu.\end{align*}
\end{proof}
For $ \delta>0$, let's denote the subset $$ \Sigma_{n,>\delta}=\left\{(x_{1},\cdots,x_{n/2})\in \Sigma_{n}:x_1^2>\delta{n\choose 2}\right \}$$ and $$\Sigma_{n,a,b,\delta}=\Sigma_{n,a,b}\setminus \Sigma_{n,>\delta}.$$
We will use Lemmas \ref{lemma5} and  \ref{lemma6} to prove 
 \begin{lem}\label{lemma7}  We have the following estimates, \\
(a) If $0<a<1<b,\ \delta>0,$ then \begin{align*}&\liminf_{n\to+\infty}\frac{1}{n^2}\ln\int_{\Sigma_{n,a,b,\delta}}Z_n^{-1}|\Delta(\mu)|^2e^{
-\sum\limits_{j=1}^{n/2}\mu_j^2/2}d\mu \geq  0.
\end{align*}(b) If $0<a<b\leq 1,\ \delta>0,$ then\begin{align*}&\liminf_{n\to+\infty}\frac{4}{n^2}\ln\int_{\Sigma_{n,a,b,\delta}}Z_n^{-1}|\Delta(\mu)|^2e^{
-\sum\limits_{j=1}^{n/2}\mu_j^2/2}d\mu \geq  1-b+\ln b.
\end{align*}(c) If $1\leq a<b,\ \delta>0,$ then\begin{align*}&\liminf_{n\to+\infty}\frac{4}{n^2}\ln\int_{\Sigma_{n,a,b,\delta}}Z_n^{-1}|\Delta(\mu)|^2e^{
-\sum\limits_{j=1}^{n/2}\mu_j^2/2}d\mu \geq  1-a+\ln a.
\end{align*}
\end{lem}\begin{proof}By \eqref{23}, we first have
\begin{align*}\int_{\Sigma_{n,>\delta}}|\Delta(\mu)|^2e^{
-\sum\limits_{j=1}^{n/2}\mu_j^2/2}d\mu&\leq \int_{\Sigma_{n}}|\Delta(\mu)|^2e^{
\mu_{1}^2/4-\sum\limits_{j=1}^{n/2}\mu_j^2/2}e^{-\delta{n\choose 2}/4}d\mu\\ &\leq 2^{2n}{Z_{n}}e^{-\delta{n\choose 2}/4}.
\end{align*}
Then, by Lemma \ref{lemma6}, we further have \begin{align*}&\int_{\Sigma_{n,a,b,\delta}}|\Delta(\mu)|^2e^{
-\sum\limits_{j=1}^{n/2}\mu_j^2/2}d\mu\\ \geq& \int_{\Sigma_{n,a,b}}|\Delta(\mu)|^2e^{-\sum\limits_{j=1}^{n/2}\mu_j^2/2}d\mu-\int_{\Sigma_{n,>\delta}}|\Delta(\mu)|^2e^{
-\sum\limits_{j=1}^{n/2}\mu_j^2/2}d\mu\\ \geq & {Z_n}\left(1-(ae^{1-a})^{\frac{n(n-1)}{4}}-(be^{1-b})^{\frac{n(n-1)}{4}}-2^{2n}e^{-\delta{n\choose 2}/4}\right).
\end{align*}Thus if $n$ is large enough, for every fixed $0<a<1<b,\ \delta>0$, using $0<ae^{1-a}<1,\ 0<be^{1-b}<1,$ we have  \begin{equation}\label{dd}\int_{\Sigma_{n,a,b,\delta}}|\Delta(\mu)|^2e^{
-\sum\limits_{j=1}^{n/2}\mu_j^2/2}d\mu \geq  {Z_n}/2,
\end{equation}which implies \begin{align*}&\liminf_{n\to+\infty}\frac{1}{n^2}\ln\int_{\Sigma_{n,a,b,\delta}}Z_n^{-1}|\Delta(\mu)|^2e^{
-\sum\limits_{j=1}^{n/2}\mu_j^2/2}d\mu \geq  0,
\end{align*}which finishes part (a).

 For every fixed $a,b,\lambda, \delta>0$ such that $0<a<1/\lambda<b$ (i.e., $0<\lambda a<1<\lambda b$), if we change variables first and then apply \eqref{dd}, we will have \begin{align*} \int_{\Sigma_{n,a,b,\delta}}|\Delta(\mu)|^2e^{-\lambda\sum\limits_{j=1}^{n/2}\mu_j^2/2}d\mu&=
\lambda^{-\frac{n(n-1)}{4}}\int_{\Sigma_{n,\lambda a,\lambda b,\lambda\delta}}|\Delta(\mu)|^2e^{-\sum\limits_{j=1}^{n/2}\mu_j^2/2}d\mu\\ &\geq
{Z_n}\lambda^{-\frac{n(n-1)}{4}}/2.
\end{align*}

If $\lambda>1$, we  have \begin{align*} \int_{\Sigma_{n,a,b,\delta}}|\Delta(\mu)|^2e^{-\sum\limits_{j=1}^{n/2}\mu_j^2/2}d\mu &\geq
e^{(\lambda-1)a{n\choose 2}/2}\int_{\Sigma_{n,a,b,\delta}}|\Delta(\mu)|^2e^{-\lambda\sum\limits_{j=1}^{n/2}\mu_j^2/2}d\mu\\ &\geq
e^{(\lambda-1)a\frac{n(n-1)}{4}}{Z_n}\lambda^{-\frac{n(n-1)}{4}}/2.
\end{align*}
Therefore, if $0<a<b\leq1,\ \delta>0,$ for every $\lambda\in(1/b,1/a) $ which is greater than 1, the above arguments imply \begin{align*}&\liminf_{n\to+\infty}\frac{4}{n^2}\ln\int_{\Sigma_{n,a,b,\delta}}Z_n^{-1}|\Delta(\mu)|^2e^{
-\sum\limits_{j=1}^{n/2}\mu_j^2/2}d\mu \geq  (\lambda-1)a-\ln \lambda.
\end{align*}Letting $ \lambda\to  (1/ a)-$, we have\begin{align*}&\liminf_{n\to+\infty}\frac{4}{n^2}\ln\int_{\Sigma_{n,a,b,\delta}}Z_n^{-1}|\Delta(\mu)|^2e^{
-\sum\limits_{j=1}^{n/2}\mu_j^2/2}d\mu \geq  1-a+\ln a.
\end{align*}Notice that for every $0\leq a<a'<b\leq1, $ we have\begin{align*}&\int_{\Sigma_{n,a,b,\delta}}Z_n^{-1}|\Delta(\mu)|^2e^{
-\sum\limits_{j=1}^{n/2}\mu_j^2/2}d\mu \geq  \int_{\Sigma_{n,a',b,\delta}}Z_n^{-1}|\Delta(\mu)|^2e^{
-\sum\limits_{j=1}^{n/2}\mu_j^2/2}d\mu,
\end{align*}and thus \begin{align*}&\liminf_{n\to+\infty}\frac{4}{n^2}\ln\int_{\Sigma_{n,a,b,\delta}}Z_n^{-1}|\Delta(\mu)|^2e^{
-\sum\limits_{j=1}^{n/2}\mu_j^2/2}d\mu\\ \geq&  \liminf_{n\to+\infty}\frac{4}{n^2}\ln\int_{\Sigma_{n,a',b,\delta}}Z_n^{-1}|\Delta(\mu)|^2e^{
-\sum\limits_{j=1}^{n/2}\mu_j^2/2}d\mu \geq  1-a'+\ln a'.
\end{align*}Letting $ a'\to b-$, we have\begin{align*}&\liminf_{n\to+\infty}\frac{4}{n^2}\ln\int_{\Sigma_{n,a,b,\delta}}Z_n^{-1}|\Delta(\mu)|^2e^{
-\sum\limits_{j=1}^{n/2}\mu_j^2/2}d\mu \geq  1-b+\ln b,
\end{align*} which finishes part (b). The proof of part (c) follows part (b) similarly and we omit the proof.

\end{proof}
\subsection{LDP in an auxiliary space}\label{ldpau}
Let's prove Proposition \ref{ldpp}. The whole proof   is separated into three parts.
\subsubsection{Lower and upper bounds}\label{boud}
We will prove the following \begin{lem}\label{lems}
\begin{align*}&\lim_{\epsilon\to 0+}\liminf_{n\to+\infty}\frac{4}{n^2}\ln \mathbb{P}(d(\gamma_n,x)<\epsilon) \geq-I(x),\\ &\lim_{\epsilon\to 0+}\limsup_{n\to+\infty}\frac{4}{n^2}\ln \mathbb{P}(d(\gamma_n,x)<\epsilon) \leq-I(x),
\end{align*}where $I(x)$ is given by \eqref{i}. 
\end{lem}
Let's first consider the lower bound. Given $x=(x_j)_{j=0}^{\infty}\in X, $ by definition we have $\sum\limits_{j=1}^{+\infty}x_j^2\leq x_0<+\infty$ and $\lim\limits_{\delta\to0+}\sum\limits_{j=1}^{+\infty}\min(x_j^2,\delta)=0 $ by monotone convergence theorem. For every $ \epsilon\in(0,1),$ there exists $k>0$ such that $\sum\limits_{j=k+1}^{+\infty}x_j^2<\epsilon^2/2.$ Let's take $ \delta\in(0,\epsilon)$ such that $\sqrt{k}\delta<\sqrt{x_0+\epsilon/2}-\sqrt{x_0}, $ then we  have \begin{lem}\label{claim1}Let $y=(y_j)_{j=0}^{\infty} \in X_0$, i.e., $y_0= \sum\limits_{j=1}^{+\infty}y_j^2$.  If $x_j<y_j<x_j+\delta$ for $1\leq j\leq k$, $y_{k+1}<\epsilon$ and $a<\sum\limits_{j=k+1}^{+\infty}y_j^2<a+\epsilon/2, $ where $a:=x_0-\sum\limits_{j=1}^{+\infty}x_j^2\geq 0 $, then $d(x,y)<\epsilon.$\end{lem}

\begin{proof}  Since $x_{k+1}^2\leq \sum\limits_{j=k+1}^{+\infty}x_j^2<\epsilon^2/2,$ thus $0\leq x_{k+1}<\epsilon.$ By assumption $0\leq y_{k+1}<\epsilon$, we have $$\sup\limits_{j\geq k+1}|x_j-y_j|\leq \sup\limits_{j\geq k+1}\max(x_j,y_j)\leq \max(x_{k+1},y_{k+1})<\epsilon,$$ where we used the fact that  the coordinate of $x,y \in X$ is decreasing. 

If we combine this with the assumption that $|x_j-y_j|<\delta<\epsilon$ for $1\leq j\leq k$, we must have \begin{equation}\label{distd}d(x,y)=\sup\limits_{j\geq 0}|x_j-y_j|\leq \max(|x_0-y_0|,\epsilon).\end{equation}Notice that $\sum\limits_{j=1}^{k}x_j^2<\sum\limits_{j=1}^{k}y_j^2<\sum\limits_{j=1}^{k}(x_j+\delta)^2, $ that $\sqrt{k}\delta<\sqrt{x_0+\epsilon/2}-\sqrt{x_0}, $ that\begin{align}\nonumber0<&\sum\limits_{j=1}^{k}(x_j+\delta)^2-\sum\limits_{j=1}^{k}x_j^2=
2\delta\sum\limits_{j=1}^{k}x_j+k\delta^2
\leq 2\delta\left(k\sum\limits_{j=1}^{k}x_j^2\right)^{\frac{1}{2}}+k\delta^2\\ \label{2}&\leq 2\delta\left(kx_0\right)^{\frac{1}{2}}+k\delta^2=(\sqrt{x_0}+\sqrt{k}\delta)^2-x_0<\epsilon/2,
\end{align}and that $a<\sum\limits_{j=k+1}^{+\infty}y_j^2<a+\epsilon/2, $ we have $\sum\limits_{j=1}^{k}x_j^2+a<\sum\limits_{j=1}^{+\infty}y_j^2=y_0<\sum\limits_{j=1}^{k}(x_j+\delta)^2+a+\epsilon/2
<\sum\limits_{j=1}^{k}x_j^2+\epsilon/2+a+\epsilon/2.$ We also have $ \sum\limits_{j=k+1}^{+\infty}x_j^2<\epsilon^2/2<\epsilon/2 $ and $\sum\limits_{j=1}^{k}x_j^2\leq \sum\limits_{j=1}^{+\infty}x_j^2=x_0-a=\sum\limits_{j=1}^{k}x_j^2+\sum\limits_{j=k+1}^{+\infty}x_j^2
<\sum\limits_{j=1}^{k}x_j^2+\epsilon/2, $ thus $x_0-\epsilon/2<\sum\limits_{j=1}^{k}x_j^2+a<y_0<\sum\limits_{j=1}^{k}x_j^2+a+\epsilon\leq x_0+\epsilon,$ i.e., $|x_0-y_0|<\epsilon$. This completes the proof by \eqref{distd}.\end{proof}

 Given $x\in X$ and $0<\delta/4<\delta<\epsilon$ defined above, recall $\gamma_n\in X_0$, by Lemma \ref{claim1} where  we replace $y$ by $\gamma_n$, for $n>2k,\ n/2\in\mathbb{Z}$, we have
 \begin{align*}& \mathbb{P}(d(\gamma_n,x)<\epsilon)\geq \mathbb{P}\bigg(x_j<{n\choose 2}^{-\frac{1}{2}}\mu_j<x_j+\delta,\ \forall\ 1\leq j\leq k;\\ &{n\choose 2}^{-\frac{1}{2}}\mu_{k+1}<\delta/4;\ a{n\choose 2}<\sum\limits_{j=k+1}^{n/2}\mu_j^2<(a+\epsilon/2){n\choose 2}\bigg).
\end{align*}

For $n$ large enough, we have $a{n\choose 2}<(a+\epsilon/4){n-2k\choose 2}$. Let $m:=n-2k$ again, and $\delta_{j}:=\dfrac{4k-j}{4k}\delta\in [\delta/2,\delta) $ for $1\leq j\leq 2k,$ then we have
\begin{align}\nonumber& \mathbb{P}(d(\gamma_n,x)<\epsilon)\geq \mathbb{P}\bigg(x_j+\delta_{2j}<{n\choose 2}^{-\frac{1}{2}}\mu_j<x_j+\delta_{2j-1},\ \forall\ 1\leq j\leq k;\\ \nonumber&\mu_{k+1}<{m\choose 2}^{\frac{1}{2}}\delta/4;\ (a+\epsilon/4){m\choose 2}<\sum\limits_{j=k+1}^{n/2}\mu_j^2<(a+\epsilon/2){m\choose 2}\bigg)=\\ \label{non} &\int_{\cap_{j=1}^{k}\{x_j+\delta_{2j}<{n\choose 2}^{-\frac{1}{2}}\mu_j<x_j+\delta_{2j-1}\}}\int_{\Sigma_{m,a+\epsilon/4,a+\epsilon/2,(\delta/4)^2}}
Z_n^{-1}|\Delta(\mu)|^2e^{-\sum\limits_{j=1}^{n/2}\mu_j^2/2}d\mu.
\end{align}By definition of $X$,  we have $x_j\geq x_{j+1}\geq 0,$ thus if $x_j+\delta_{2j}<{n\choose 2}^{-\frac{1}{2}}\mu_j<x_j+\delta_{2j-1} $ for $1\leq j\leq k$, we will have $$\mu_j-\mu_{j+1}> (\delta_{2j}-\delta_{2j+1}){n\choose 2}^{\frac{1}{2}}=\frac{\delta}{4k}{n\choose 2}^{\frac{1}{2}}$$ for $1\leq j< k$ and $$\mu_k> \delta_{2k}{n\choose 2}^{\frac{1}{2}}=\frac{\delta}{2}{n\choose 2}^{\frac{1}{2}}.$$ If $\mu_{k+1}<{m\choose 2}^{\frac{1}{2}}\delta/4<{n\choose 2}^{\frac{1}{2}}\delta/4 $, then we have $$\mu_k-\mu_{k+1}> \frac{\delta}{2}{n\choose 2}^{\frac{1}{2}}-\frac{\delta}{4}{n\choose 2}^{\frac{1}{2}}=\frac{\delta}{4}{n\choose 2}^{\frac{1}{2}}\geq \frac{\delta}{4k}{n\choose 2}^{\frac{1}{2}}.$$
 Therefore,  for $1\leq l\leq k,$ we must have  $$\frac{\Delta(\mu_{>l-1})}{\Delta(\mu_{>l})}=\prod\limits_{j=l+1}^{ n/2}(\mu_l^2-\mu_j^2)\geq (\mu_l^2-\mu_{l+1}^2)^{n/2-l}\geq (\mu_l-\mu_{l+1})^{n-2l}\geq \left(\frac{\delta}{4k}{n\choose 2}^{\frac{1}{2}}\right)^{n-2l},$$ and hence, $$\frac{\Delta(\mu)}{\Delta(\mu_{>k})}=\prod\limits_{l=1}^{ k}\frac{\Delta(\mu_{>l-1})}{\Delta(\mu_{>l})}\geq \prod\limits_{l=1}^{ k}\left(\frac{\delta}{4k}{n\choose 2}^{\frac{1}{2}}\right)^{n-2l}=\left(\frac{\delta}{4k}{n\choose 2}^{\frac{1}{2}}\right)^{nk-k(k+1)}. $$
By \eqref{2}, we have $${n\choose 2}^{-1}\sum\limits_{j=1}^{k}\mu_j^2\leq\sum\limits_{j=1}^{k}(x_j+\delta)^2\leq \sum\limits_{j=1}^{k}x_j^2+\epsilon/2\leq \sum\limits_{j=1}^{+\infty}x_j^2+\epsilon/2= x_0-a+\epsilon/2.$$ Therefore,  for $n$ large enough, we can further estimate \eqref{non} as \begin{align*}& Z_n\mathbb{P}(d(\gamma_n,x)<\epsilon)\geq \int_{\cap_{j=1}^{k}\{x_j+\delta_{2j}<{n\choose 2}^{-\frac{1}{2}}\mu_j<x_j+\delta_{2j-1}\}}\int_{\Sigma_{m,a+\epsilon/4,a+\epsilon/2,(\delta/4)^2}}
\\ &\left(\frac{\delta}{4k}{n\choose 2}^{\frac{1}{2}}\right)^{2(nk-k(k+1))}|\Delta(\mu_{>k})|^2e^{-\sum\limits_{j=k+1}^{n/2}\mu_j^2/2-{n\choose 2}(x_0-a+\epsilon/2)/2}d\mu\\ &=\left(\left(\frac{\delta}{4k}\right)^{2}{n\choose 2}\right)^{nk-k(k+1)}\left(\prod_{j=1}^k(\delta_{2j-1}-\delta_{2j}){n\choose 2}^{\frac{1}{2}}\right)e^{-{n\choose 2}(x_0-a+\epsilon/2)/2}\\ &\times \int_{\Sigma_{m,a+\epsilon/4,a+\epsilon/2,(\delta/4)^2}}|\Delta(\mu_{>k})|^2
e^{-\sum\limits_{j=k+1}^{n/2}\mu_j^2/2}d\mu_{>k}\\ &=\left[\left(\left(\frac{\delta}{4k}\right)^{2}{n\choose 2}\right)^{nk-k(k+1/2)}e^{-(x_0-a+\epsilon/2)\frac{n(n-1)}{4}}Z_m\right] \\ &\times \left[Z_m^{-1}\int_{\Sigma_{m,a+\epsilon/4,a+\epsilon/2,(\delta/4)^2}}|\Delta(\mu_{>k})|^2
e^{-\sum\limits_{j=k+1}^{n/2}\mu_j^2/2}d\mu_{>k}\right].
\end{align*}Here, we used the fact that $(\delta_{2j-1}-\delta_{2j}){n\choose 2}^{\frac{1}{2}}=\frac{\delta}{4k}{n\choose 2}^{\frac{1}{2}}. $ Since $m=n-2k,\ Z_n=(\pi/2)^{\frac{n}{4}}\prod\limits_{j=0}^{ n/2-1}(2j)!,$ we have\begin{align*} {Z_n}/Z_m= \prod\limits_{j=1}^{ k}(\pi/2)^{\frac{1}{2}}(n-2j)!\leq\left((\pi/2)^{\frac{1}{2}}n!\right)^k.
\end{align*}Using $n!\leq n^n$ and ${n\choose 2}\geq n>0, $ we have ${Z_n}/Z_m \leq (\pi/2)^{\frac{k}{2}}n^{kn} $, thus \begin{align*}&\left(\left(\frac{\delta}{4k}\right)^{2}{n\choose 2}\right)^{nk-k(k+1/2)}Z_m/Z_n\\ \geq& \left(\left(\frac{\delta}{4k}\right)^{2}{n\choose 2}\right)^{nk-k(k+1/2)}(\pi/2)^{-\frac{k}{2}}n^{-kn}\\ \geq& \left(\frac{\delta}{4k}\right)^{2nk-2k(k+1/2)}{n\choose 2}^{-k(k+1/2)}(\pi/2)^{-\frac{k}{2}}.
\end{align*}Therefore, we have \begin{align*} \mathbb{P}(d(\gamma_n,x)<\epsilon)&\geq \left(\frac{\delta}{4k}\right)^{2nk-2k(k+1/2)}{n\choose 2}^{-k(k+1/2)}(\pi/2)^{-\frac{k}{2}}e^{-(x_0-a+\epsilon/2)\frac{n(n-1)}{4}}\\ & \times Z_m^{-1}\int_{\Sigma_{m,a+\epsilon/4,a+\epsilon/2,(\delta/4)^2}}|\Delta(\mu_{>k})|^2
e^{-\sum\limits_{j=k+1}^{n/2}\mu_j^2/2}d\mu_{>k}.
\end{align*} Hence,  we have \begin{align*}&\liminf_{n\to+\infty}\frac{4}{n^2}\ln \mathbb{P}(d(\gamma_n,x)<\epsilon)\geq -(x_0-a+\epsilon/2)\\&+\liminf_{m\to+\infty}\frac{4}{m^2}\ln \int_{\Sigma_{m,a+\epsilon/4,a+\epsilon/2,(\delta/4)^2}}Z_m^{-1}|\Delta(\mu_{>k})|^2
e^{-\sum\limits_{j=k+1}^{n/2}\mu_j^2/2}d\mu_{>k}.
\end{align*}If $a\geq 1,$ by part (c) of Lemma \ref{lemma7}, we have
  \begin{align*}&\liminf_{n\to+\infty}\frac{4}{n^2}\ln \mathbb{P}(d(\gamma_n,x)<\epsilon)\\ \geq & -(x_0-a+\epsilon/2)+1-(a+\epsilon/4)+\ln(a+\epsilon/4).
\end{align*}Since for $ \epsilon'\in(0,\epsilon)$, we have $\mathbb{P}(d(\gamma_n,x)<\epsilon)\geq \mathbb{P}(d(\gamma_n,x)<\epsilon')$, thus \begin{align*}&\liminf_{n\to+\infty}\frac{4}{n^2}\ln \mathbb{P}(d(\gamma_n,x)<\epsilon)\geq\liminf_{n\to+\infty}\frac{4}{n^2}\ln \mathbb{P}(d(\gamma_n,x)<\epsilon')\\ \geq & -(x_0-a+\epsilon'/2)+1-(a+\epsilon'/4)+\ln(a+\epsilon'/4),
\end{align*}letting $\epsilon'\to0+ $, we obtain\begin{align*}&\liminf_{n\to+\infty}\frac{4}{n^2}\ln \mathbb{P}(d(\gamma_n,x)<\epsilon)\\ \geq & -(x_0-a)+1-a+\ln a=-x_0+1+\ln a.
\end{align*}Similarly, if $ 0\leq a<1$, then for $0<\epsilon<1-a$, we have $0<a+\epsilon/4<a+\epsilon/2<1$. Now by Lemma \ref{lemma7} again, we have \begin{align*}&\liminf_{n\to+\infty}\frac{4}{n^2}\ln \mathbb{P}(d(\gamma_n,x)<\epsilon)\geq  -(x_0-a+\epsilon/2)\\&+\liminf_{m\to+\infty}\frac{4}{m^2}\ln \int_{\Sigma_{m,a+\epsilon/4,a+\epsilon/2,(\delta/4)^2}}Z_m^{-1}|\Delta(\mu_{>k})|^2
e^{-\sum\limits_{j=k+1}^{n/2}\mu_j^2/2}d\mu_{>k}\\ \geq&-(x_0-a+\epsilon/2)+1-(a+\epsilon/2)+\ln(a+\epsilon/2).
\end{align*}If $0<\epsilon<1,\ 0<\epsilon'<\min(1-a,\epsilon)$, then\begin{align*}&\liminf_{n\to+\infty}\frac{4}{n^2}\ln \mathbb{P}(d(\gamma_n,x)<\epsilon)\geq\liminf_{n\to+\infty}\frac{4}{n^2}\ln \mathbb{P}(d(\gamma_n,x)<\epsilon')\\ \geq & -(x_0-a+\epsilon'/2)+1-(a+\epsilon'/2)+\ln(a+\epsilon'/2),
\end{align*}letting $\epsilon'\to0+ $, we obtain\begin{align*}&\liminf_{n\to+\infty}\frac{4}{n^2}\ln \mathbb{P}(d(\gamma_n,x)<\epsilon) \geq-x_0+1+\ln a.
\end{align*}Therefore, for $x=(x_j)_{j=0}^{\infty}\in X,\ a=x_0-\sum\limits_{j=1}^{+\infty}x_j^2,\ \epsilon\in(0,1) $, we always have the lower bound\begin{align}\label{3}&\liminf_{n\to+\infty}\frac{4}{n^2}\ln \mathbb{P}(d(\gamma_n,x)<\epsilon) \geq-x_0+1+\ln a,
\end{align}in the sense that $\ln 0=-\infty.$  Recall  the definition of $a$ in Lemma \ref{claim1},  we have$$a=J(x)=x_0-\sum\limits_{j=1}^{+\infty}x_j^2\geq 0.$$ This implies the lower bound in Lemma \ref{lems} if we define $I(x):=x_0-1-\ln J(x)$.

Now we consider the upper bound. For $A,B\in\mathbb{R},\ k\in\mathbb{Z},\ k>0,$ let's define $$G(x)=(A-B)\sum\limits_{j=1}^{k}x_j^2+Bx_0,\,\,\, x=(x_j)_{j=0}^{\infty}\in X.$$  Then $G$ is continuous in $X$ and $$G(x)=A\sum\limits_{j=1}^{k}x_j^2+B\sum\limits_{j=k+1}^{+\infty}x_j^2 \,\,\,\mbox{if}\,\,\, x\in X_0.$$ Now for every $\delta>0,$ there exists $ \epsilon\in(0,1)$ depending only on $x, A,B,k,\delta$ such that $ G(y)>G(x)-\delta$ for $y\in X, $ $d(x,y)<\epsilon$. By definition of $\gamma_n\in X_0$, we further have \begin{align*}&G(\gamma_n)={n\choose 2}^{-1}A\sum\limits_{j=1}^{k}\mu_j^2+{n\choose 2}^{-1}B\sum\limits_{j=k+1}^{n/2}\mu_j^2.
\end{align*}If $A, B<1/2,$ by \eqref{12} in Lemma \ref{lemma5}, we have\begin{align*}& \mathbb{P}(d(\gamma_n,x)<\epsilon)\leq \mathbb{P}(G(\gamma_n)>G(x)-\delta)\leq e^{-{n\choose 2}(G(x)-\delta)}\mathbb{E}e^{{n\choose 2}G(\gamma_n)}\\&=e^{-{n\choose 2}(G(x)-\delta)}\mathbb Ee^{A\sum\limits_{j=1}^k\mu_j^2+B\sum\limits_{j=k+1}^{n/2}\mu_j^2}\\&\leq e^{-{n\choose 2}(G(x)-\delta)} 2^{nk}(1-2A)^{-k(n-k-\frac{1}{2})}(1-2B)^{-(\frac{n}{2}-k)(\frac{n-1}{2}-k)},
\end{align*}which implies\begin{align*}&\limsup_{n\to+\infty}\frac{4}{n^2}\ln \mathbb{P}(d(\gamma_n,x)<\epsilon)\leq  -2(G(x)-\delta)-\ln(1-2B)\\&=-2(A-B)\sum\limits_{j=1}^{k}x_j^2-2Bx_0+2\delta-\ln(1-2B).
\end{align*}As $\limsup\limits_{n\to+\infty}\dfrac{4}{n^2}\ln \mathbb{P}(d(\gamma_n,x)<\epsilon) $ is an increasing function of $ \epsilon,$ for every $A<1/2,\ B<1/2,\ \delta>0,\ k\in\mathbb{Z},\ k>0,$ we have\begin{align*}&\lim_{\epsilon\to 0+}\limsup_{n\to+\infty}\frac{4}{n^2}\ln \mathbb{P}(d(\gamma_n,x)<\epsilon)\\ \leq&-2(A-B)\sum\limits_{j=1}^{k}x_j^2-2Bx_0+2\delta-\ln(1-2B).
\end{align*}Letting $A\to (1/2)-,\ k\to+\infty,\ \delta\to 0+$, we have\begin{align*}&\lim_{\epsilon\to 0+}\limsup_{n\to+\infty}\frac{4}{n^2}\ln \mathbb{P}(d(\gamma_n,x)<\epsilon)\\ \leq&-(1-2B)\sum\limits_{j=1}^{+\infty}x_j^2-2Bx_0-\ln(1-2B)\\=&(1-2B)a-x_0-\ln(1-2B).
\end{align*}
Therefore, we can further find the upper bound of the last line by choosing $B=(1-1/a)/2$ if $a>0$ and $B\to -\infty$ if $a=0$ (i.e., $x\in X_0$), and thus we have\begin{align}\label{4}&\lim_{\epsilon\to 0+}\limsup_{n\to+\infty}\frac{4}{n^2}\ln \mathbb{P}(d(\gamma_n,x)<\epsilon) \leq1-x_0+\ln a,
\end{align} which finishes the upper bound in Lemma \ref{lems}.





\subsubsection{Compactness}\label{cpm}
Let's recall that the rate function $I(x)$ is  good if its level sets $\{x|I(x)\leq t\} $ are compact. 
 We first give the following compactness criterion,
 
 \begin{lem}\label{claim2} The level sets $A_t:=\{x=(x_j)_{j=0}^{\infty}\in X|x_0\leq t\} $ are compact.\end{lem}

\begin{proof}Since the function $F(x)=x_0$ is continuous in $X,$ the level sets $A_t=\{x=(x_j)_{j=0}^{\infty}\in X|x_0\leq t\}$ are closed and $A_t=\emptyset$ for $t<0.$ If $t\geq 0,$ given a sequence $\{x^k=(x_j^k)_{j=0}^{\infty}\}\subset A_t,$ we have $0\leq x_0^k\leq t,$ $(x_1^k)^2\leq \sum\limits_{j=1}^{+\infty}(x_j^k)^2\leq x_0\leq t$ and $0\leq x_j^k\leq x_1^k\leq t^{\frac{1}{2}} $ for $j\geq 1.$ Now we can find a subsequence $\{x^{(k)}=(x_j^{(k)})_{j=0}^{\infty}\}\subset A_t$ and $x^{(0)}=(x_j^{(0)})_{j=0}^{\infty} $ such that $\lim\limits_{k\to+\infty}x_j^{(k)}= x_j^{(0)}$ for $j\geq 0.$ By Fatou's lemma and the definitions of $X$ and $A_t$, we have $x^{(0)}\in A_t $ and $\lim\limits_{j\to+\infty}x_j^{(0)}=0.$ Now for $k,l\in\mathbb{Z},\ k,l\geq 0$, we have $\sup\limits_{j\geq l+1}|x_j^{(k)}-x_j^{(0)}|\leq \sup\limits_{j\geq l+1}\max(x_j^{(k)},x_j^{(0)})\leq \max(x_{l+1}^{(k)},x_{l+1}^{(0)})\leq x_{l+1}^{(0)}+|x_{l+1}^{(k)}-x_{l+1}^{(0)}|;$ for $0\leq j\leq l$, we have $|x_j^{(k)}-x_j^{(0)}|\leq  x_{l+1}^{(0)}+|x_{j}^{(k)}-x_{j}^{(0)}|.$ Thus $ d(x^{(k)},x^{(0)})=\sup\limits_{j\geq 0}|x_j^{(k)}-x_j^{(0)}|\leq x_{l+1}^{(0)}+\max\limits_{0\leq j\leq l+1}|x_{j}^{(k)}-x_{j}^{(0)}|$ and\begin{align*}0\leq&\limsup_{k\to+\infty}d(x^{(k)},x^{(0)})\leq x_{l+1}^{(0)}+\limsup_{k\to+\infty}\max\limits_{0\leq j\leq l+1}|x_{j}^{(k)}-x_{j}^{(0)}|\\=&x_{l+1}^{(0)}+\max\limits_{0\leq j\leq l+1}\limsup_{k\to+\infty}|x_{j}^{(k)}-x_{j}^{(0)}|=x_{l+1}^{(0)}.
\end{align*}Letting $l\to+\infty$, we have $\limsup\limits_{k\to+\infty}d(x^{(k)},x^{(0)})=0, $ which means $x^{(k)}\to x^{(0)} $ in $X$ and $A_t$ is compact. This completes the proof.\end{proof}

Since $I$ is lower semicontinuous, the level sets $\{x|I(x)\leq t\} $ are closed. For $x=(x_j)_{j=0}^{\infty}$,  we have $0\leq J(x)\leq x_0$, and thus $I(x)=x_0-1-\ln J(x)\geq x_0-1-\ln x_0=x_0/2+(x_0/2-1-\ln (x_0/2))-\ln 2\geq x_0/2-\ln 2$. Thus if $I(x)\leq t$, then $x_0\leq 2(t+\ln 2),$ which implies $\{x|I(x)\leq t\}\subseteq A_{2(t+\ln 2)}.$ By Lemma \ref{claim2}, $A_{2(t+\ln 2)}$ is compact, thus the level sets $\{x|I(x)\leq t\} $ are compact. Therefore, the rate function $I(x)$ is good.

\subsubsection{Exponential tightness}

We say that the sequence $Y_1,Y_2,\cdots$ is exponentially
tight if for any $E>0$, there exists a compact set $K_E\subset X$ such that \begin{align*}&\limsup_{n\to+\infty}\frac{1}{a_n}\ln \mathbb{P}(Y_n\not\in K_E) <-E.
\end{align*}
Regarding the exponentially tight measures, we have (see Appendix D in  \cite{AGZ}),
\begin{lem}\label{lem1}Let $(Y_n)_{n>0,n\in\mathbb{Z}}$ be a sequence of random variables taking values
in some Polish space $V$. Suppose that it is exponentially tight. If there exists
a lower semicontinuous function $I:V\to[0,+\infty]$, such that for all $x\in V$ the
following estimates of small ball probabilities hold\begin{align*}&\lim_{\epsilon\to 0+}\limsup_{n\to+\infty}\frac{1}{a_n}\ln \mathbb{P}(Y_n\in B(x,\epsilon)) \leq-I(x),\\ &\lim_{\epsilon\to 0+}\liminf_{n\to+\infty}\frac{1}{a_n}\ln \mathbb{P}(Y_n\in B(x,\epsilon)) \geq-I(x).
\end{align*}Then $(Y_n)_{n>0,n\in\mathbb{Z}}$ satisfies LDP with rate function $I(x)$.\end{lem}

By Lemma \ref{lem1} and the results in \S \ref{boud} and \S \ref{cpm}, Proposition \ref{ldpp} follows once we prove that the sequence of random variables $(\gamma_n)_{n>0,n\in 2\mathbb Z}$ is exponentially tight.  




Since the function $F(x)=x_0$ is continuous in $X$ and $F(\gamma_n)=(\gamma_n)_0={n\choose 2}^{-1}\sum\limits_{j=1}^{n/2}\mu_j^2 $,
 taking $k=0,\ b=1/4$ in \eqref{12}, we have  \begin{align*}& \mathbb{P}(\gamma_n\not\in A_t)=\mathbb{P}(F(\gamma_n)>t)\leq e^{-{n\choose 2}t/4}\mathbb Ee^{{n\choose 2}F(\gamma_n)/4}\\=&e^{-{n\choose 2}t/4}\mathbb Ee^{\sum\limits_{j=1}^{n/2}\mu_j^2/4} \leq e^{-{n\choose 2}t/4}(1-2/4)^{-\frac{n}{2}\frac{n-1}{2}}.
\end{align*}Then \begin{align*}&\limsup_{n\to+\infty}\frac{4}{n^2}\ln \mathbb{P}(\gamma_n\not\in A_t) \leq -t/2-\ln(1-2/4)=-t/2+\ln 2.
\end{align*}For any $E>0$, let's choose $t=2(E+1)>0$ and $K_E:=A_t\subset X$, then\begin{align*}&\limsup_{n\to+\infty}\frac{4}{n^2}\ln \mathbb{P}(\gamma_n\not\in K_E)\leq -t/2+\ln 2<-t/2+1=-E.
\end{align*}By Lemma \ref{claim2}, $K_E$ is compact, and thus $(\gamma_n)_{n>0,n\in\mathbb Z}$ is exponentially tight. This will complete  the proof of Proposition \ref{ldpp}.
\subsection{Proof of Theorem \ref{ldp2}}\label{ssss}
Now we are ready to prove Theorem \ref{ldp2}.  

As explained in \S\ref{ldppi}, let's define the map $ \varphi:X\to M_1(\mathbb{R})$ via its Fourier transform
   \eqref{definf}, then by definition of $\gamma_n$, we must have $ \varphi(\gamma_n)=\rho_n$. There are three more properties we need to prove.  First,   $\varphi(x) $ is a Borel probability
measure. In fact, $\varphi(x) $ is the density of the random variable $Y=a_0+\sum\limits_{j=1}^{+\infty}x_j a_j,$ where $(a_j)_{j=0}^{\infty}$ are independent random variables such that $\mathbb{P}(a_j=1)=\mathbb{P}(a_j=-1)=1/2 $ for $j>0,$ and $a_0$ is a Gaussian random variable with mean $0$ and variance $J(x).$

Secondly, the map $ \varphi$ is continuous. To show this, 
we need the following fundamental lemma which indicates that the pointwise convergence of the Fourier transform convergence implies the convergence in $(M_1(\mathbb{R}), d_{BL})$ \cite{AGZ, Br}, 

\begin{lem}\label{lem2} If $ \mu_n,\mu\in M_1(\mathbb{R}),\ \lim\limits_{n\to+\infty}\widehat{\mu_n}(s)=\widehat{\mu}(s)$ for every $s\in \mathbb{R},$ then $\mu_n\to\mu$ in $ M_1(\mathbb{R}),\ i.e.\lim\limits_{n\to+\infty}d_{BL}(\mu_n,\mu)=0. $\end{lem}

Now, for $x=(x_j)_{j=0}^{\infty}\in X,\ J(x)=x_0- \sum_{j=1}^{+\infty}x_j^2,$ we have\begin{align*}&\widehat{\varphi(x)}(s)=e^{-J(x)s^2/2}\prod_{j=1}^{+\infty}\cos sx_j=e^{-x_0s^2/2}\prod_{j=1}^{+\infty}(e^{(sx_j)^2/2}\cos sx_j).
\end{align*}For $x^k=(x_j^k)_{j=0}^{\infty}\in X$ such that $x^k\to x^0$ in $X$, we have $\lim\limits_{k\to+\infty}x_j^k= x_j^0$ for every fixed $j\geq 0$ and \begin{align}\label{5}\lim\limits_{k\to+\infty}e^{-x_0^ks^2/2}= e^{-x_0^0s^2/2},\ \lim\limits_{k\to+\infty}e^{(sx_j^k)^2/2}\cos sx_j^k= e^{(sx_j^0)^2/2}\cos sx_j^0.\end{align} Since $\lim\limits_{t\to0}t^{-2}\ln\cos t=-1/2, $ then for every $ \delta>0$, there exists $ \epsilon\in(0,1)$ such that $|t^2/2+\ln\cos t|\leq t^2\delta$ for $|t|<\epsilon.$ Since $\lim\limits_{j\to+\infty}x_j^0=0,$ for every fixed $s\in\mathbb{R}$, there exists $l>0$ such that $|sx_l^0|<\epsilon,$ then there exists $k_0>0$ such that $|sx_l^k|<\epsilon,\ x_0^k<x_0^0+1$ for $k>k_0,$ thus $|sx_j^k|\leq |sx_l^k|<\epsilon$ for $ j\geq l,\ k>k_0, $ and for $k>k_0,\ l'>l$ we have\begin{align*}&\left|\ln\prod_{j=l'}^{+\infty}(e^{(sx_j^k)^2/2}\cos sx_j^k)\right| \leq \sum_{j=l}^{+\infty}\left|(sx_j^k)^2/2+\ln\cos sx_j^k\right|\leq \sum_{j=l}^{+\infty}(sx_j^k)^2\delta\\ &\leq s^2\delta\sum_{j=1}^{+\infty}(x_j^k)^2\leq s^2\delta x_0^k\leq s^2\delta (x_0^0+1),
\end{align*}which implies the uniform convergence of the infinite product\begin{align*}&e^{-x_0^ks^2/2}\prod_{j=1}^{+\infty}(e^{(sx_j^k)^2/2}\cos sx_j^k),\ \ k\geq 0.
\end{align*}By \eqref{5}, we have\begin{align*}&\lim_{k\to+\infty}e^{-x_0^ks^2/2}\prod_{j=1}^{+\infty}(e^{(sx_j^k)^2/2}\cos sx_j^k)=e^{-x_0^0s^2/2}\prod_{j=1}^{+\infty}(e^{(sx_j^0)^2/2}\cos sx_j^0),
\end{align*}this gives $\lim\limits_{k\to+\infty}\widehat{\varphi(x^k)}(s)=\widehat{\varphi(x^0)}(s) $ for every fixed $s\in \mathbb{R}.$ Therefore, by Lemma \ref{lem2}, we conclude the continuity of $ \varphi$.

In the end,  the map $ \varphi$ is injective. In fact, the second moment of the probability measure $\varphi(x)$ reads \begin{align*}&\langle{\varphi(x)},\lambda^2\rangle=J(x)+\sum_{j=1}^{+\infty}x_j^{2}=x_0,
\end{align*}thus $x_0$ can be determined uniquely by $\varphi(x). $ Now we prove that  $x_j$ can be determined inductively by $\varphi(x). $ Let $f_0(s)=\widehat{\varphi(x)}(s)$, then $x_1=\pi/(2\inf\limits\{t>0|f_0(t)=0\})$ and $x_1=0$ if $f_0(t)\neq 0$ for all $t\in\mathbb{R}.$ Once $f_{k-1}$ and $x_{k}$ are determined, let $f_k(s)=f_{k-1}(s)/\cos sx_k$ and extend $f_k$ to be a continuous function for $s\in\mathbb{R},$ then $x_{k+1}=\pi/(2\inf\limits\{t>0|f_k(t)=0\})$ and $x_{k+1}=0$ if $f_k(t)\neq 0$ for all $t\in\mathbb{R}.$ In this way, we can determine $x_j,\ j>0$ only using $f_0.$ Thus $ \varphi$ is injective.

Now we can give the LDP of $\rho_n$ by the following Contraction Principle \cite{AGZ},

\begin{lem}\label{lem3} Let $(Y_n)_{n>0,n\in\mathbb{Z}}$ be a sequence of random variables taking values
in some Polish space $X$. Let $ \varphi:X\to V$ be continuous and injective, $V$ is also a Polish space. If  $(Y_n)_{n>0,n\in\mathbb{Z}}$ satisfies LDP with speed $a_n$, going to
infinity with $n$, and rate function $I$ which is good, then $(\varphi(Y_n))_{n>0,n\in\mathbb{Z}}$ satisfies LDP with speed $a_n$ and good rate function $\widetilde{I}$ such that $\widetilde{I}(x)=I(\varphi^{-1} x)$ if $x\in\varphi(X)$ and $\widetilde{I}(x)=+\infty$ if $x\not\in\varphi(X).$\end{lem}



By Lemma \ref{lem3} and the continuity and injectivity of $\varphi$ we proved above, $(\rho_{n})_{n>0,n\in 2\mathbb{Z}}$ will satisfy the LDP in $( M_1(\mathbb R), d_{BL})$ with speed $a_n=n^2/4$ and good rate function $\widetilde{I}(x)$ such that $\widetilde{I}(x)=I(\varphi^{-1} x)$ if $x\in\varphi(X)$ and $\widetilde{I}(x)=+\infty$ if $x\not\in\varphi(X)$. This completes the proof of Theorem \ref{ldp2}.

\section{Concentration of measure theorem for $q_n\geq 3$}\label{cmp12}
Now we discuss the concentration of
measure theorem for  $\rho_n$ (defined in \eqref{emp}) of eigenvalues of the Gaussian SYK model for  general $q_n\geq 3$. 
 \subsection{Notations and basic properties}Let's first recall some notations and basic properties in \cite{FTD1} regarding the Majorana fermions.  For a set $A=\{i_1, i_2,\cdots, i_{m}\}$$  \subseteq\{1,2,\cdots,n\}$, $ 1\leq i_1<i_2< \cdots <i_m \leq n,$ we denote $$\Psi_A:=\psi_{i_1}\cdots \psi_{i_{m}} \,\,\,\mbox{and}\,\,\,\Psi_A:=I\,\,\mbox{if}\,\,A=\emptyset.$$ 
We will need the following properties,\\
  \textcircled 1 Given a set $A\subseteq\{1,2,\cdots,n\}$,   $$\mbox{$\Tr\Psi_A=0 $ and $\Psi_A\neq \pm I $ are always true for $A\neq \emptyset$}.$$ \textcircled 2 For $A,B\subseteq\{1,2,\cdots,n\}$, then  $$\mbox{$\Psi_A=\pm\Psi_B$ if and only if $A=B$.}$$
 \textcircled 3  $$\mbox{$\Psi_A\Psi_B=\pm\Psi_{A\triangle B}$  where $A\triangle B:=(A\setminus B)\cup (B\setminus A).$}$$ 

\subsection{Proof of Lemma \ref{lemlip}}

\begin{proof} Recall the notation \eqref{simple}, we may consider the random matrices $H$ as  functions $H(x)$ which 
maps $x:=(J_R)_{R\in I_n}\in \mathbb R^{{n\choose q_n}}$ to the space of  $L_n\times L_n$ Hermitian matrices which is equipped
with the Hilbert-Schmidt norm $$\|A\|_{H.S.}^2:=\Tr (AA^*)=\Tr (A^2).$$
For $x=(J_R)_{R\in I_n},x'=(J_R')_{R\in I_n}\in \R^{{n\choose q_n}}, $ let's write $H:=H(x)$ and $H':=H(x').$ Then\begin{align*}&\|H-H'\|_{H.S.}^2=\Tr[(H-H')^2]\\&=(-1)^{[q_n/2]}{n\choose q_n}^{-1}\sum_{R\in I_n}\sum_{R'\in I_n}(J_R-J_R')(J_{R'}-J_{R'}')\Tr(\Psi_R\Psi_{R'}).
\end{align*}By properties \textcircled 1 \textcircled 2 \textcircled 3 above, if $R\neq R' $, then $R\triangle R'\neq \emptyset$ and $\Tr(\Psi_R\Psi_{R'})=\pm\Tr(\Psi_{R\triangle R'})=0. $ If $R= R' $, then by the anticommutative property \eqref{anti},  we must have $\Psi_R\Psi_{R'}=\Psi_R^2=(-1)^{[q_n/2]}I$ and $\Tr(\Psi_R\Psi_{R'})=(-1)^{[q_n/2]}L_n.$ It follows that\begin{equation}\label{dds}\|H-H'\|_{H.S.}^2=L_n{n\choose q_n}^{-1}\sum_{R\in I_n}(J_R-J_R')^2=L_n{n\choose q_n}^{-1}\|x-x'\|^2,
\end{equation}and thus the map $x\mapsto H(x)$ is $L_n^{1/2}{n\choose q_n}^{-1/2} $-Lipschitz.

Now we consider the map $H\mapsto\langle f,\rho_n\rangle$ where $f$ is Lipschitz. Let $(\la_j)_{1\leq j\leq L_n}$ be eigenvalues of $H$ and $(\la_j')_{1\leq j\leq L_n}$ be eigenvalues of $H'$ such that $\la_j\geq \la_{j+1},\ \la_j'\geq \la_{j+1}'$ for $1\leq j<L_n.$ 
By definition $\langle f, \rho_n\rangle=L_n^{-1}\sum\limits_{j=1}^{L_n}f(\la_j)$, we have\begin{align}\label{tt}&\nonumber|\langle f, \rho_n\rangle-\langle f, \rho_n'\rangle|=L_n^{-1}\left|\sum\limits_{j=1}^{L_n}(f(\la_j)-f(\la_j'))\right|\\ &\leq \frac{\|f'\|_{L^\infty(\mathbb R)}}{L_n}\sum\limits_{j=1}^{L_n}|\la_j-\la_j'|\leq \|f'\|_{L^\infty(\mathbb R)}\sqrt{L_n^{-1}\sum\limits_{j=1}^{L_n}|\la_j-\la_j'|^2},
\end{align}where we used the fact that $f$ is Lipschitz. The Hoffman-Wielandt inequality \cite{AGZ} further yields \begin{equation}\label{hw}\sqrt{L_n^{-1}\sum\limits_{j=1}^{L_n}|\la_j-\la_j'|^2}\leq L_n^{-1/2}\|H-H'\|_{H.S.},
\end{equation}thus the map $H\mapsto\langle f,\rho_n\rangle$ is $L_n^{-1/2}\|f'\|_{L^\infty(\mathbb R)}$-Lipschitz. Therefore,
if we combine \eqref{dds}\eqref{tt}\eqref{hw},  when $f$ is 1-Lipschitz, we  have \begin{align*}&|\langle f, \rho_n\rangle-\langle f, \rho_n'\rangle|\leq L_n^{-1/2}\|H-H'\|_{H.S.}\\&=L_n^{-1/2}\sqrt{L_n{n\choose q_n}^{-1}}\|x-x'\|={n\choose q_n}^{-1/2}\|x-x'\|,
\end{align*} i.e.,  the map $x \mapsto H \mapsto \langle f,\rho_n\rangle $ is ${n\choose q_n}^{-1/2} \|f'\|_{L^\infty(\mathbb R)}$-Lipschitz, which finishes (a).
By triangle inequality and   definition of the bounded Lipschitz metric \eqref{md}, after taking supremum over all 1-Lipschitz functions, we further have \begin{align*}|d_{BL}(\rho_n,\rho)-d_{BL}(\rho_n',\rho)|\leq d_{BL}(\rho_n,\rho_n')\leq{n\choose q_n}^{-1/2}\|x-x'\|,
\end{align*}this completes the proof of (b).\end{proof}





\subsection{Proof of Theorem \ref{cmt}}
Now we are ready to prove Theorem \ref{cmt}.
\begin{proof} 
By part (b) of Lemma \ref{lemlip}, $d_{BL}(\rho_n,\rho_{\infty})$ is $L $-Lipschitz with Lipschitz constant $L={n\choose q_n}^{-1/2} $. Therefore, by the concentration of measure theorem for Gaussian vectors \eqref{P}, for any $t>0$, we will have\begin{equation}\label{ddss}\mathbb{P}(|d_{BL}(\rho_n,\rho_{\infty})-\e d_{BL}(\rho_n,\rho_{\infty})|>t)\leq Ce^{-ct^2/L^2}=Ce^{-c{n\choose q_n}t^2}.
\end{equation}
For the empirical measure of eigenvalues $\rho_n$ defined in \eqref{emp} and its limit $\rho_{\infty} $ as proved   in \cite{FTD1}, for any 1-Lipschitz function $f$,  we have \begin{align*}|\langle f, \rho_n\rangle-\langle f, \rho_{\infty}\rangle|\leq |\langle f, \rho_n\rangle|+|\langle f, \rho_{\infty}\rangle|\leq 2\|f\|_{L^{\infty}}\leq 2.
\end{align*}Thus  $0\leq d_{BL}(\rho_n,\rho_{\infty})\leq 2. $ The fact that $\rho_n\to\rho_\infty$ almost surely (which is one of the main results in \cite{FTD1}) implies $$\lim\limits_{n\to+\infty}d_{BL}(\rho_n,\rho_{\infty})=0,\ \ \text{a.s.}$$ Therefore, by the dominated convergence theorem, 
we  have 
 $$\lim\limits_{n\to+\infty}\e d_{BL}(\rho_n,\rho_{\infty})=0.$$ 
Thus for every $a>0,$  there exists $N_0=N_0(a)>0$ such that $\e d_{BL}(\rho_n,\rho_{\infty})\leq a/2 $ for $n>N_0$. Thus if $n>N_0$,  by \eqref{ddss}, we have \begin{align*} \mathbb{P}(d_{BL}(\rho_n,\rho_{\infty})>a)\leq  \mathbb{P}(|d_{BL}(\rho_n,\rho_{\infty})-\e d_{BL}(\rho_n,\rho_{\infty})|>a/2)\leq  Ce^{-c{n\choose q_n}(a/2)^2}.
\end{align*}If $n\leq N_0$, we have \begin{align*} \mathbb{P}(d_{BL}(\rho_n,\rho_{\infty})>a)\leq 1=ee^{-1} \leq  Ce^{-{n\choose q_n}2^{-n}} \leq  Ce^{-{n\choose q_n}2^{-N_0}}.
\end{align*}Therefore, we always have $$ \mathbb{P}(d_{BL}(\rho_n,\rho_{\infty})>a)\leq Ce^{-c(a){n\choose q_n}} $$ for $c(a)=\min(c\cdot(a/2)^2,2^{-N_0(a)}),$ which completes the proof Theorem \ref{cmt}.\end{proof}

\end{document}